\documentclass[11pt]{article}

\usepackage[utf8]{inputenc}
\usepackage[english]{babel}
\usepackage{amsmath}
\usepackage{nicefrac}
\usepackage{amsthm}
\usepackage{amsfonts}
\usepackage{amssymb}
\usepackage{verbatim}
\usepackage{color}
\usepackage[arrow, matrix, curve]{xy}
\usepackage{bbm}
\usepackage{eqparbox}
\usepackage[numbers]{natbib}
\usepackage{stmaryrd}
\usepackage{amssymb}
\usepackage{mathrsfs}
\usepackage{pictexwd,dcpic}

\DeclareMathSymbol{\leq}{\mathrel}{symbols}{20}
   
\DeclareMathSymbol{\geq}{\mathrel}{symbols}{21}

\newtheoremstyle{WreschTheoremstyle} 
                        {1.5em}    
                        {2.5em}    
                        {}         
                        {}         
                        {\bfseries}
                        {}        
                        {\newline} 
                        {\raisebox{0.6em}{\thmname{#1}\thmnumber{#2}\thmnote{ (#3)}}}

\newcommand{\R}{\mathbb{R}}

\newcommand{\N}{\mathbb{N}}

\newcommand{\e}{\varepsilon}

\newtheorem{Theorem}{Theorem}[section]
\newtheorem{Corollary}[Theorem]{Corollary}
\newtheorem{Lemma}[Theorem]{Lemma}
\newtheorem{Remark}[Theorem]{Remark}

\numberwithin{equation}{section}

\makeatletter
\newcommand{\customlabel}[1]{%
     \stepcounter{ref}%
   \protected@write
\@auxout{}{\string\newlabel{#1}{{\thesatz.\arabic{ref}}{\thepage}{\thesatz.\arabic{ref}}{#1}{}}}%
   \hypertarget{#1}{\thesatz.\arabic{ref}}%
}
\makeatother

\newenvironment{sciabstract}{\begin{quote}}{\end{quote}}

\newcounter{lastnote}

\title{Weak-coupling limit for ergodic environments}
\newcommand{\pdftitle} {Weak-coupling limit for ergodic environments}
\newcommand{\pdfauthor}{Martin Friesen}

\author{
Martin Friesen\footnote{Fakult\"at f\"ur Mathematik und Naturwissenschaften, Bergische Universit\"at Wuppertal, Gaußstraße 20, 42119 Wuppertal, Germany, friesen@math.uni-wuppertal.de}\\
Yuri Kondratiev\footnote{Department of Mathematics, Bielefeld University, Germany, kondrat@math.uni-bielefeld.de}
}

\usepackage[plainpages=false,pdfpagelabels=true,bookmarks=true,pdfauthor={\pdfauthor},
pdftitle={\pdftitle}]{hyperref}

\makeatletter
\def\HyPsd@CatcodeWarning#1{}
\makeatother

\begin{document}

\maketitle

\begin{sciabstract}\textbf{Abstract:}
 The main aim of this work is to establish an averaging principle for a wide class of interacting particle systems in the continuum.
 This principle is an important step in the analysis of Markov evolutions and is usually applied for the associated semigroups related to  backward Kolmogorov equations,  c.f. \cite{KURTZ73}.
 Our approach is based on the study of  forward Kolmogorov  equations (a.k.a.  Fokker-Planck equations).
 We  describe a system evolving as a Markov process on the space of finite configurations, whereas its rates depend on the actual state of another (equilibrium) process
 on the space of locally finite configurations. We will show that ergodicity of the environment process implies the averaging principle for the solutions of the 
 coupled Fokker-Planck equations. 
\end{sciabstract}


\section{Introduction}

This work is devoted to the study of interacting particle systems with a continuous state space, say $\R^d$.
Particles are supposed to be indistinguishable completely determined by their
positions denoted by $x \in \R^d$.
Particular models are used in various fields such as physics, chemistry, ecology, medicine and even social sciences, 
where it is usually supposed that particles are subject to some Markovian dynamics including elementary events such as birth, deaths and jumps.
A rigorous study of these models by stochastic differential equations
is, e.g., performed in \cite{HK90, FM04, BEZ151, FM16, X18, FRS18}
while analytic tools have been used in \cite{KOL03, KOL06, MU-FA04, EW03}.
Note that all models mentioned above assume that the total number of particles is finite at any moment of time, i.e. they are modeled on the state space 
of locally finite configurations
\[
 \Gamma_0 = \{ \eta \subset \R^d \ | \ |\eta| < \infty \},
\]
where $|A|$ denotes the number of elements in the set $A$.
Such space can be equipped with a natural topology such that it becomes a locally compact Polish space. 

In this work we study such particle system in the presence of an environment
described by another particle system on the space of locally finite configurations
\[
 \Gamma = \{\gamma \subset \R^d \ | \ |\gamma \cap K| < \infty \ \ \text{ for all compacts } K \subset \R^d \}.
\]
In order to distinguish between $\Gamma$ and $\Gamma_0$, we use $\gamma$ for elements in $\Gamma$, while $\eta,\xi,\zeta$ belong to $\Gamma_0$. We endow $\Gamma$ with the smallest topology such that, for any continuous function $f: \R^d \longrightarrow \R$ with compact support,
$\gamma \longmapsto \sum_{x \in \gamma}f(x)$ is continuous. It can be shown that $\Gamma$ is a Polish space, see \cite{KK06}.
Note that, in contrast to $\Gamma_0$, this space is not locally compact.

Let us describe the general form of the dynamics (system and environment) studied in this work. For a fixed configuration of the environment $\gamma \in \Gamma$,
dynamics of the system is supposed to be given by the heuristic Markov operator
\begin{align}\label{EQ:11}
 (L^S(\gamma)F)(\eta) = \sum \limits_{\xi \subset \eta}\int \limits_{\Gamma_0}(F(\eta \backslash \xi \cup \eta) - F(\eta))K(\gamma,\xi,\eta,d\zeta),
\end{align}
where $K(\gamma,\xi,\eta,d\zeta) \geq 0$ describes the infinitesimal transition rate 
for the elementary Markov event $\eta \longmapsto \eta \backslash \xi \cup \zeta$.
Such transition rate should satisfy some reasonable assumptions, and can be seen as a continuum analogue of the Kolmogorov matrix known from the theory of Markov chains on countable state spaces. Denote by $L^E$ the Markov generator for the Markov evolution of the environment on $\Gamma$, i.e. an unbounded operator acting on a suitable class of functions $F: \Gamma \longrightarrow \R$.
The corresponding Markov process 
for the joint evolution, system and environment, can be formally obtained 
from the (backward) Kolmogorov equation
\begin{align}\label{EQ:00}
 \frac{d}{dt}F_t = \left(L^S + L^E\right)F_t, \ \ F_t|_{t=0} = F_0, \ \ t \geq 0,
\end{align}
where $F_t: \Gamma_0 \times \Gamma \longrightarrow \R$,
$L^S = L^S(\gamma)$ is given by \eqref{EQ:11} and acts only on the variable $\eta$,
while $L^E$ acts only in the $\gamma$ variable.
Note that such description is only heuristic, i.e. in this generality the
corresponding Markov process does not need to exist
and, moreover, equation \eqref{EQ:00} does not need to have any solution at all. 

Ignoring for a moment the construction of solutions to \eqref{EQ:00}, let us denote by $(\eta_t,\gamma_t) \in \Gamma_0 \times \Gamma$
the Markov process obtained by formally solving \eqref{EQ:00}. 
Since $L^S$ is assumed to depend on $\gamma$, it is clear that $\eta_t$ 
alone is not a Markov process on $\Gamma_0$.
However, for different regimes of parameters one may still hope that
the system process $\eta_t$ is at least close to a Markovian process in some reasonable sense.
From the mathematical point of view the latter one 
results in the requirement to find a certain scaling $(\eta_t^{\e}, \gamma_t^{\e})$,
$\e > 0$, and show that $\eta_t^{\e} \longrightarrow \overline{\eta}_t$
for $\e \to 0$, where $\overline{\eta}_t$ is a Markovian process.
Therefore, such an approximation is a 
particular case of Markovian limits as discussed in \cite{S80}.

If the environment process is ergodic with invariant measure $\mu$, 
then the weak-coupling limit, 
which is a particular case of so-called random evolution framework, see, e.g., \cite{PINSKY, SHS},
can be used to approximate $\eta_t$ by a Markov process obtained from the averaged Markov operator
\[
 (\overline{L}F)(\eta) = \sum \limits_{\xi \subset \eta}\int \limits_{\Gamma_0}( F(\eta \backslash \xi \cup \zeta) - F(\eta)) \overline{K}(\xi,\eta,d\zeta),
\]
where $\overline{K}(\xi,\eta, d\zeta) = \int_{\Gamma}K(\gamma,\xi,\eta, d\zeta)d\mu(\gamma)$.
More precisely, consider, for $\e \in (0,1)$, the scaled (backward) Kolmogorov equation
\begin{align}\label{EQ:04}
 \frac{d}{dt}F_t^{\e} = \left( L^S + \frac{1}{\e}L^E\right)F_t^{\e}, \ \ F_t^{\e}|_{t=0} = F_0, \ \ t \geq 0,
\end{align}
where the initial condition $F_0 = F_0(\eta)$ is supposed to be independent of the variable $\gamma$.
Then one expects that $F_t^{\e} \longrightarrow \overline{F}_t$ as $\e \to 0$ and $\overline{F}_t = \overline{F}_t(\eta)$ solves
\[
 \frac{d}{dt}\overline{F}_t = \overline{L}\overline{F}_t, \ \ \overline{F}_t|_{t= 0} = F_0(\eta), \ \ t \geq 0.
\]
Such a formal scheme was established for various situations based on the theory of stochastic equations 
or on a detailed study of the (backward) Kolmogorov equation.
However, at present there exist no methods for the rigorous study of 
\eqref{EQ:04}.

In this work we propose another approach to study the weak-coupling limit 
$\e \to 0$. Namely, let $\mathcal{L}_{\mu} := L^1(\Gamma_0 \times \Gamma,\lambda \otimes \mu)$ and consider the 
evolution of densities obtained from the Fokker-Planck equation
\begin{align}\label{EQ:01}
 \frac{d}{dt} \rho_t^{\e}(\eta,\gamma) = \left( L^S + \frac{1}{\e}L^E\right)^* \rho_t^{\e}, \  \ \rho_t^{\e}|_{t= 0} = \rho_0 \in L^1(\Gamma_0,\lambda),
\end{align}
where $\left( L^S + \frac{1}{\e}L^E\right)^*$ denotes the adjoint operator to $L^S + \frac{1}{\e}L^E$.
Note that $\rho_t^{\e}$ describes the one-dimensional distributions of the corresponding Markov process $(\eta_t^{\e}, \gamma_t^{\e})$,
provided it exists. 
Hence we seek to prove that $\rho_t^{\e} \longrightarrow \overline{\rho}_t$ in $\mathcal{L}_{\mu}$ as $\e \to 0$, 
and, moreover, show that $\overline{\rho}_t \in L^1(\Gamma_0, \lambda)$ is independent of $\gamma$ satisfying the Fokker-Planck equation
\begin{align}\label{EQ:03}
 \frac{d}{dt}\overline{\rho}_t = \overline{L}^*\overline{\rho}_t, \ \ \overline{\rho}_t|_{t=0} = \rho_0,
\end{align}
where $\overline{L}^*$ denotes the adjoint operator to $\overline{L}$.
In \cite{FK18} we have shown by a different approach that for more specific models of birth-and-death type the restriction that the system dynamics evolves in $\Gamma_0$ can be dropped,
i.e. a similar result was obtained for spatial birth-and-death processes on the larger state space $\Gamma \times \Gamma$.
Contrary to this, the result obtained in this work applies to a significantly larger class of dynamics.

This work is organized as follows.
In Section 2 we discuss the construction of Markovian particle systems
on $\Gamma_0$ with rates independent of the environment. 
Our main result of this work is then formulated and proved in Section 3.
Finally, a particular example is discussed in the last Section of this work.

\section{Some results for finite particle systems}

\subsection{Space of finite configurations}
Set $\Gamma_0^{(0)} = \{ \emptyset\}$ and, for $n \geq 1$, 
$\Gamma_0^{(n)} = \{ \eta \subset \mathbb{R}^d\ | \ \ |\eta| = n\}$. Then
\[
 \Gamma_0 = \{ \eta \subset \mathbb{R}^d\ | \ \ |\eta| < \infty \} = \bigsqcup \limits_{n = 0}^{\infty}\Gamma_0^{(n)},
\]
where $|A|$ denotes the number of elements in the set $A \subset \R^d$.
Let us describe the topology used on $\Gamma_0$.
Denote by $\widetilde{(\mathbb{R}^d)^n}$ the collection of all $(x_1, \dots, x_n) \in (\mathbb{R}^d)^n$ with $x_i \neq x_j$ for $i \neq j$, and set
\[
 \mathrm{sym}_n: \widetilde{(\mathbb{R}^d)^n} \to \Gamma_0^{(n)}, \ (x_1, \dots, x_n) \longmapsto \{x_1, \dots, x_n\}.
\]
A set $A \subset \Gamma_0$ is said to be open iff 
$\mathrm{sym}_n^{-1}(A \cap \Gamma_0^{(n)}) \subset \widetilde{(\mathbb{R}^d)^n}$ is open for all $n \geq 0$ in the relative topology on $(\R^d)^n$.
It can be shown that $\Gamma_0$ equipped with this topology is a locally compact Polish space \cite{BDKMO18}. Moreover, the corresponding Borel-$\sigma$-algebra
$\mathcal{B}(\Gamma_0)$ is generated by cylinder sets 
$\{ \eta \in \Gamma_0 \ | \ |\eta \cap \Lambda| = n \}$,
where $n \geq 0$ and $\Lambda \subset \R^d$ is compact.

The Lebesgue-Poisson measure $\lambda$ on $\Gamma_0$ is defined by the relation
\[
 \int \limits_{\Gamma_0}G(\eta) d\lambda(\eta) = G(\{\emptyset \}) + \sum \limits_{n=1}^{\infty}\frac{1}{n!}\int \limits_{(\R^d)^n}G(\{x_1,\dots, x_n\})dx_1\dots dx_n,
\]
where $G$ is any Borel-measurable non-negative function on $\Gamma_0$.
This measure satisfies, for any measurable function $G: \Gamma_0 \times \Gamma_0 \longrightarrow \R$, the integration by parts formula
\begin{align}\label{IBP}
 \int \limits_{\Gamma_0}\sum \limits_{\xi \subset \eta}G(\xi,\eta \backslash \xi) d\lambda(\eta)
 = \int \limits_{\Gamma_0}\int \limits_{\Gamma_0}G(\xi,\eta)d\lambda(\xi)d\lambda(\eta),
\end{align}
provided one side of the equality is finite for $|G|$, see \cite[Appendix]{F17}.

\subsection{Markovian dynamics on $\Gamma_0$}
In this section we briefly describe Markovian dynamics on $\Gamma_0$
consisting of elementary Markovian events such as
\[
 \eta \longmapsto \eta \backslash \xi \cup \zeta, \qquad \zeta \subset \R^d \backslash (\eta \backslash \xi), \ \ \xi \subset \eta.
\]
Such events should occur with infinitesimal transition 
rate $K: \Gamma_0 \times \Gamma_0 \times \Gamma_0 \longrightarrow \mathbb{R}_+$ satisfying
\begin{enumerate}
 \item[(K)] The map $(\xi, \eta, \zeta) \longmapsto K(\xi, \eta, \zeta)$ is jointly Borel-measurable and
 \[
  \int \limits_{\Gamma_0}K(\xi,\eta,\zeta)d\lambda(\zeta) < \infty, \ \ \forall \eta, \xi \in \Gamma_0.
 \]
\end{enumerate}
Denote by $BM(\Gamma_0)$ the Banach space of all bounded measurable functions equipped with the supremum norm.
For $F \in BM(\Gamma_0)$ define 
\begin{align}\label{KOLMOGOROV}
 (AF)(\eta) = \sum \limits_{\xi \subset \eta}\int \limits_{\Gamma_0}(F(\eta \backslash \xi \cup \zeta) - F(\eta))K(\xi, \eta, \zeta)d\lambda(\zeta), \ \ \eta \in \Gamma_0.
\end{align}
Note that $AF$ is pointwisely well-defined but, in general, does not need to be bounded.
Such operator is supposed to describe a pure-jump Markov process on $\Gamma_0$, which may have an explosion.
This can be seen from the following representation given below.
For $\eta \in \Gamma_0$ and $A \in \mathcal{B}(\Gamma_0)$ we define a transition kernel 
\[
 Q(\eta, A) := \sum \limits_{\xi \subset \eta} \int \limits_{\Gamma_0}\mathbbm{1}_A(\eta \backslash \xi \cup \zeta)K(\xi, \eta, \zeta)d\lambda(\zeta)
\]
describing the infinitesimal rate from state $\eta$ to the set $A$.
The total transition rate is therefore given by
\[
 q(\eta):= Q(\eta, \Gamma_0) = \sum \limits_{\xi \subset \eta}\int \limits_{\Gamma_0}K(\xi,\eta, \zeta)d\lambda(\eta).
\]
Note that (K) implies that $q(\eta)$ is finite for each $\eta \in \Gamma_0$. 
The action of the operator $A$ can be then rewritten to 
\begin{align}\label{EQ:05}
 (AF)(\eta) = -q(\eta)F(\eta) + \int \limits_{\Gamma_0}F(\xi)Q(\eta,d\xi) = \int \limits_{\Gamma_0}(F(\xi) - F(\eta))Q(\eta,d\xi).
\end{align}
A construction and some properties of the corresponding minimal (sub-)Markov transition function 
$P: \mathbbm{R}_+ \times \Gamma_0 \times \mathcal{B}(\Gamma_0) \longrightarrow [0,1]$ was studied in \cite{FELLER40, MU-FA04, FMS14}.
Based on the theory of Lyapunov functions, the corresponding transition semigroup, in particular the Feller property, was recently studied in \cite{FRIESEN15}.
Below we provide a construction of the (sub-)Markov transition function and the associated semigroup based on the theory of sub-stochastic semigroups.

Denote by $\mathcal{M}(\Gamma_0)$ the Banach space of signed Borel measures with finite total variation
\[
 \Vert \nu \Vert = |\nu|(\Gamma_0) = \nu_+(\Gamma_0) + \nu_-(\Gamma_0), \ \ \nu \in \mathcal{M}(\Gamma_0),
\]
where $\nu_+,\nu_-$ denote the Hahn-Jordan decomposition of $\nu$ and $|\nu| := \nu_+ + \nu_-$. The (formally) adjoint operator to $A$ should act on $\mathcal{M}(\Gamma_0)$. Using the representation \eqref{EQ:05}, it is clear that it should be given by
\[
 (\mathcal{A}\nu)(C) = - \int \limits_{C}q(\eta)\nu(d\eta) + \int \limits_{\Gamma_0}Q(\eta,C)\nu(d\eta), \ \ C \in \mathcal{B}(\Gamma_0)
\]
equipped with the domain
\[
 D(\mathcal{A}) = \left\{ \nu \in \mathcal{M}(\Gamma_0)\ \bigg| \ \int \limits_{\Gamma_0}q(\eta)|\nu|(d \eta) < \infty \right\}.
\]
A strongly continuous semigroup $S(t)$ on $\mathcal{M}(\Gamma_0)$ is called sub-stochastic,
if $S(t)\nu \geq 0$ and $\| S(t)\nu \| \leq \| \nu \|$ whenever $0 \leq \nu \in \mathcal{M}(\Gamma_0)$.
Then we obtain the following.
\begin{Theorem}
 Suppose that (K) is satisfied. 
 \begin{enumerate}
  \item[(a)] The operator $(\mathcal{A}, D(\mathcal{A}))$ is well-defined and has an extension $(\mathcal{G}, D(\mathcal{G}))$ on $\mathcal{M}(\Gamma_0)$ which 
  is the generator of a sub-stochastic semigroup $S(t)$.
  Moreover, this semigroup is minimal in the following sense: Let $(\widetilde{S}(t))_{t \geq 0}$ be another sub-stochastic semigroup on $\mathcal{M}(\Gamma_0)$ 
  with generator being an extension of $(\mathcal{G}, \mathcal{D}(\mathcal{G}))$. 
  Then $\widetilde{S}(t)\nu \leq S(t)\nu$
  for all $0 \leq \nu \in \mathcal{M}(\Gamma_0)$ and $t \geq 0$.
  \item[(b)] There exists a (sub-)Markovian transition function $P$ such that
  \begin{align}\label{BKE}
   S(t)\nu(C) = \int \limits_{\Gamma_0}P(t,\eta,C)\nu(d\eta),\ \ t \geq 0, \ \ C \in \mathcal{B}(\Gamma_0).
  \end{align}
  \item[(c)] For each $\nu \in \mathcal{M}(\Gamma_0)$ and $F \in BM(\Gamma_0)$ the duality
  \[
   \int \limits_{\Gamma_0}S(t)F(\eta) \nu(d\eta) = \int \limits_{\Gamma_0}F(\eta) (S(t)\nu)(d\eta), \ \ t \geq 0
  \]
  holds, where $S(t)^*F$ is given by 
  \begin{align}\label{BKE:1}
   S(t)F(\eta) = \int \limits_{\Gamma_0}F(\xi)P(t,\eta,d\xi),\ \ t \geq 0
  \end{align}
  \item[(d)] $S(t)$ leaves the space $L^1(\Gamma_0,\lambda) \subset \mathcal{M}(\Gamma_0)$ invariant.
  Its restriction to $L^1(\Gamma_0,\lambda)$ is again a strongly continuous semigroup.
 \end{enumerate}
\end{Theorem}
\begin{proof}
 First observe that the multiplication operator $(-q, D(\mathcal{A}))$ 
 given by $-q\nu(C) = -\int_{C}q(\eta)\nu(d\eta)$ generates a positive analytic semigroup of contractions given by
 \[
  (e^{-tq}\nu)(C) = \int \limits_{C}e^{-tq(\eta)}\nu(d\eta), \ \ \nu \in \mathcal{M}(\Gamma_0).
 \]
 Next observe that $(B, D(\mathcal{A}))$ given by
 \[
  (B\nu)(C) = \int \limits_{\Gamma_0}Q(\eta,C)\nu(d\eta), \ \ C \in \mathcal{B}(\Gamma_0)
 \]
 is well-defined, positive and satisfies 
 \begin{align}\label{EQ:09}
  B\nu(\Gamma_0) = \int \limits_{\Gamma_0}Q(\eta, \Gamma_0)\nu(d\eta) = \int \limits_{\Gamma_0}q(\eta)\nu(d\eta), \ \ 0 \leq \nu \in \mathcal{D}(\mathcal{A}).
 \end{align}
 Hence assertion (a) is a consequence of \cite[Theorem 2.1]{TV06}.
 Assertion (b) follows from \cite[Section 5]{TV06}, 
 while property (c) can be directly deduced from \eqref{BKE} and \eqref{BKE:1}.
 It remains to prove assertion (d). For $a > 0$ and $\nu \in \mathcal{M}(\Gamma_0)$ define 
 \[
  R(a)\nu(C) = \int \limits_{C}\frac{1}{a + q(\eta)}\nu(d\eta), \ \ C \in \mathcal{B}(\Gamma_0),
 \]
 which implies that
 \[
  BR(a)\nu(C) = \int \limits_{\Gamma_0}\frac{Q(\eta,C)}{a + q(\eta)}\nu(d\eta), \ \  C \in \mathcal{B}(\Gamma_0).
 \]
 It follows from \cite{TV06} that the resolvent of $(\mathcal{G},D(\mathcal{G}))$ satisfies
 \[
  (a - \mathcal{G})^{-1}\nu = \lim \limits_{r \nearrow 1}R(a)\sum \limits_{n=0}^{\infty}r^n (BR(a))^n \nu, \ \ \nu \in \mathcal{M}(\Gamma_0),
 \]
 where the convergence is with respect to the total variation norm.
 Next observe that $L^1(\Gamma_0,\lambda)$ is closed in $\mathcal{M}(\Gamma_0)$ such that, for each $g \in L^1(\Gamma_0,\lambda)$, one has
 \[
  \| g \|_{L^1(\Gamma_0,\lambda)} =  \int \limits_{\Gamma_0}|g(\eta)|d\lambda(\eta) = \| g \lambda \|_{\mathcal{M}(\Gamma_0)}.
 \]
 Hence it suffices to show that $R(a)$ and $BR(a)$ leave $L^1(\Gamma_0, \lambda)$ invariant. 
 It is immediate that $R(a)L^1(\Gamma_0,\lambda) \subset L^1(\Gamma_0, \lambda)$. Next let $\nu = g\lambda$ with $g \in L^1(\Gamma_0,\lambda)$,
 and take $C \in \mathcal{B}(\Gamma_0)$ with $\lambda(C) = 0$. Then
 \begin{align*}
  (BR(a)\nu)(C) &= \int \limits_{\Gamma_0}Q(\eta,C) \frac{g(\eta)}{a + q(\eta)} d\lambda(\eta)
  \\ &= \int \limits_{\Gamma_0} \sum \limits_{\xi \subset \eta} \int \limits_{\Gamma_0}\mathbbm{1}_C(\eta \backslash \xi \cup \zeta)K(\xi, \eta, \zeta)d\lambda(\zeta)\frac{g(\eta)}{a + q(\eta)} d\lambda(\eta)
  \\ &= \int \limits_{\Gamma_0}\int \limits_{\Gamma_0}\int \limits_{\Gamma_0}\mathbbm{1}_C(\eta \cup \zeta) K(\xi, \eta \cup \xi, \zeta)\frac{g(\eta \cup \xi)}{a + q(\eta \cup \xi)} d\lambda(\zeta)d\lambda(\xi)d\lambda(\eta)
  \\ &= \int \limits_{\Gamma_0}\int \limits_{\Gamma_0} \mathbbm{1}_C(\eta)\sum \limits_{\zeta \subset \eta} K(\xi,\eta \cup \xi \backslash \zeta, \zeta) \frac{g(\eta \cup \xi \backslash \zeta)}{a + q(\eta \cup \xi \backslash \zeta)} d\lambda(\xi)d\lambda(\eta)
      = 0
 \end{align*}
 where we have used \eqref{IBP} twice.
\end{proof}
The semigroup $S(t)$ is called stochastic, if $S(t)\nu \geq 0$ and $\| S(t)\nu \| = \| \nu \|$ whenever $0 \leq \nu \in \mathcal{M}(\Gamma_0)$.
This is equivalent to the requirement that $P(t,\eta,\Gamma_0) = 1$ for all $t,\eta$.
It is worthwhile to mention that without any further assumptions the semigroup $S(t)$ might be not stochastic, 
i.e. $P(t,\xi, \Gamma_0) < 1$ may occur for some $t > 0$ and $\xi \in \Gamma_0$.
Sufficient conditions for $S(t)$ being stochastic can be found in \cite{MU-FA04, TV06}.
We have the following simple characterization of stochasticity. Other related results are given in \cite{ALK09}.
\begin{Corollary}\label{CHARACTERIZATIONSTOCH}
 Suppose that (K) is satisfied and let $S(t)$ be the semigroup constructed above.
 Then $S(t)$ is stochastic if and only if its generator $(\mathcal{G},D(\mathcal{G}))$ is the closure of $(\mathcal{A}, D(\mathcal{A}))$.
\end{Corollary}
\begin{proof} 
 Suppose that $(\mathcal{G}, D(\mathcal{G}))$ is the closure of $(\mathcal{A}, D(\mathcal{A}))$. By \eqref{EQ:09} we obtain $\mathcal{A}\nu(\Gamma_0) = 0$
 for $0 \leq \nu \in D(\mathcal{A}))$. This yields, by approximation,
 $\mathcal{G}\nu(\Gamma_0) = 0$ for $0 \leq \nu \in D(\mathcal{G}))$.
 Hence, for $0 \leq \nu \in D(\mathcal{G})$, we obtain
 \[
  \frac{d}{dt}\| S(t)\nu \| = \frac{d}{dt}S(t)\nu(\Gamma_0) = \mathcal{G}S(t)\nu(\Gamma_0) = 0.
 \]
 This shows that $S(t)$ is stochastic. 
 
 Conversely, suppose that $S(t)$ is stochastic. 
 Take $0 \leq \nu \in \mathcal{M}(\Gamma_0)$ and observe that by $S(t) \nu \geq 0$ and $S(t)\nu = \nu + \mathcal{G}\int_{0}^{t}S(s)\nu ds$ we have
 \[
  \nu(\Gamma_0) = \| \nu \| = \| S(t)\nu \| = S(t)\nu(\Gamma_0) = \nu(\Gamma_0) + \left( \mathcal{G}\int \limits_{0}^{t}S(s)\nu ds \right)(\Gamma_0), \qquad t > 0.
 \]
 Hence we obtain, for $0 \leq \nu \in D(\mathcal{G})$,
 \[
  0 = \left(\mathcal{G}\frac{1}{t}\int \limits_0^t S(s)\nu ds\right)(\Gamma_0) 
  = \left(\frac{1}{t}\int\limits_0^{t}S(s)\mathcal{G}\nu ds\right)(\Gamma_0)
  \longrightarrow \mathcal{G}\nu(\Gamma_0), \qquad t \to 0,
 \]
 i.e. $\mathcal{G}\nu(\Gamma_0) = 0$,
 where we have used that $\nu \longmapsto \nu(\Gamma_0)$ is continuous
 in the total variation norm.
 Using also that $\mathcal{A}\nu(\Gamma_0) = 0$ for $0 \leq \nu \in D(\mathcal{A})$,
 the assertion follows from \cite[Corollary 3.6]{ALK09}. 
\end{proof} 
The following is a particular case of \cite[Proposition 5.1]{TV06}.
\begin{Remark}
 Suppose that (K) is satisfied and assume that there exists a measurable function $V: \Gamma_0 \longrightarrow \R_+$,
 and constants $c,b > 0$ such that
 \begin{align}\label{EQ:06}
  \int \limits_{\Gamma_0}\left( V(\xi) - V(\eta) \right)Q(\eta,d\xi) \leq c(1 + V(\eta)) - \e q(\eta), \ \ \eta \in \Gamma_0.
 \end{align}
 Then $S(t)$ is stochastic and leaves
 \[
  \mathcal{M}_V(\Gamma_0) = \{ \nu \in \mathcal{M}(\Gamma_0) \ | \ \| \nu\|_V := \int \limits_{\Gamma_0}(1 + V(\eta))d\nu(\eta) < \infty \}
 \]
 invariant. Moreover, the restriction of $S(t)$ onto $\mathcal{M}_V(\Gamma_0)$ is strongly continuous.
\end{Remark}
We close this section with another sufficient condition for $S(t)$ to be stochastic due to \cite[Part I, Theorem, 2.25]{MU-FA04}.
\begin{Remark}\label{REMARK:00}
 Suppose that (K) is satisfied. Moreover, assume that 
 \begin{enumerate}
  \item[(i)] There a measurable function $V: \Gamma_0 \longrightarrow \R_+$ and a constant $c > 0$ such that 
  \[
   \int \limits_{\Gamma_0}\left( V(\xi) - V(\eta) \right)Q(\eta,d\xi) \leq cV(\eta), \ \ \eta \in \Gamma_0.
  \]
  \item[(ii)] There exists a sequence of Borel sets $(E_n)_{n \in \N} \subset \Gamma_0$ with $E_n \subset E_{n+1}$ and $\bigcup_{n \in \N} E_n = \Gamma_0$ such that
  \[
   \sup \limits_{\eta \in E_n}q(\eta) < \infty, \ \forall n \in \N, \qquad
   \lim \limits_{n \to \infty} \inf \limits_{\eta \not \in E_n}V(\eta) = \infty.
  \]
 \end{enumerate}
 Then $S(t)$ is stochastic.
\end{Remark} 

\section{The main result}

\subsection{Description of the environment}
The following is our main conditions on the enviroment:
\begin{enumerate}
 \item[(E)] There exists a Borel probability measure $\mu$ on $\Gamma$ and a positive semigroup of contractions $T^E(t)$ on $L^1(\Gamma,\mu)$, 
 which is assumed to be $L^1$-ergodic, i.e.,  for each $R \in L^1(\Gamma,\mu)$ 
 \begin{align}\label{EQ:10}
  \| T^E(t)R - P_{\mu}R \|_{L^1(\Gamma,\mu)} \longrightarrow 0, \ \ t \to \infty,
 \end{align}
 where $P_{\mu}R = \int_{\Gamma}R(\gamma)d\mu(\gamma)$ denotes the average of $R$ with respect to $\mu$.
\end{enumerate}
Such condition has the following interpretation. The environment has an equilibrium measure $\mu$ and,
if the environment is in the initial state $R d\mu$, where $R \in L^1(\Gamma, \mu)$, 
then the time evolution is given by $R_td\mu$ with $R_t = T^E(t)R$. Since, in addition, one has $T^E(t)R \longrightarrow P_{\mu}R$ in $L^1(\Gamma,\mu)$, 
the evolution of densities is ergodic on $L^1(\Gamma,\mu)$.
The following result is classical in the theory of Dirichlet forms.
It will be used to provide sufficient examples for condition (E).
\begin{Theorem}\cite[Theorem 1.4.1]{D89}
 Let $\mu$ be a Borel probability measure on $\Gamma$ and 
 let $T_2^E(t)$ be a symmetric Markov semigroup on $L^2(\Gamma,\mu)$, i.e.
 a strongly continuous semigroup with $T^E(t)1 = 1$ satisfying 
 \[
  \int \limits_{\Gamma}T_2^E(t)R \cdot H d\mu = \int \limits_{\Gamma}R \cdot T_2^E(t)H d\mu, \qquad R,H \in L^2(\Gamma, \mu),
 \]
 and $0 \leq T_2^E(t)R \leq 1$ whenever $0 \leq R \leq 1$. 
 Then $T_2^E(t)$ leaves $L^1(\Gamma,\mu) \cap L^{\infty}(\Gamma,\mu)$ 
 invariant and has, for $p \in [1,\infty)$, a unique extension $T_p^E(t)$ onto $L^p(\Gamma,\mu)$ being a positive and strongly continuous contraction semigroup. 
 These extensions satisfy
 \[
  T_p^E(t)R = T_q^E(t)R, \qquad R \in L^p(\Gamma,\mu) \cap L^q(\Gamma,\mu), \ \ 1 \leq p \leq q < \infty
 \]
 and if $\frac{1}{p} + \frac{1}{q} = 1$, then $T_p^E(t)^* = T_q^E(t)$ with
 $T_{\infty}^E(t) := T_1^E(t)^*$.
\end{Theorem} 
Based on the theory of Dirichlet forms, 
equilibrium gradient diffusions on $\Gamma$ were studied in \cite{AKR98, AKR98b}.
Equilibrium Glauber dynamics were then studied in \cite{KL05}.
For both examples it was shown that grand canonical Gibbs measures are invariant measures and the corresponding symmetric Markov semigroup $T_2^E(t)$ on $L^2(\Gamma,\mu)$ was constructed. Moreover, it was shown that this semigroup
is ergodic on $L^2(\Gamma, \mu)$, i.e.
\[
 \| T^E_2(t)R - P_{\mu}R \|_{L^2(\Gamma,\mu)} \longrightarrow 0, \qquad t \to \infty, \ \ R \in L^2(\Gamma,\mu).
\]
Let $T^E(T)$ be the strongly continous semigroup on $L^1(\Gamma, \mu)$. Then
\[
 \| T^E(t)R - P_{\mu}R \|_{L^1(\Gamma,\mu)} \leq \| T^E_2(t)R - P_{\mu}R \|_{L^2(\Gamma,\mu)}
\]
for all $R \in L^1(\Gamma,\mu) \cap L^2(\Gamma,\mu)$, i.e. \eqref{EQ:10} holds 
on a dense set of functions. Since $T^E(t)$ is a contraction operator,
by approximation it also holds for all $R \in L^1(\Gamma,\mu)$.

\subsection{Description of the system}
The system is modelled by a Markov process on $\Gamma_0$
having generator similarly to the one from Section 2.
Moreover, we suppose that its rates depend, in addition, on the configuration
of the environment. More precisely, let $L^S$ be for
any bounded measurable function $F = F(\eta,\gamma)$ given by
\begin{align}\label{EQ:02}
 (L^SF)(\eta,\gamma) = \sum \limits_{\xi \subset \eta}\int \limits_{\Gamma_0}\left( F(\eta \backslash \xi \cup \zeta, \gamma) - F(\eta, \gamma)\right)K(\gamma,\xi,\eta,\zeta)d\lambda(\zeta),
\end{align}
where $K(\gamma,\xi,\eta,\zeta)$ is supposed to satisfy
\begin{enumerate}
 \item[(S)] $K: \Gamma \times \Gamma_0 \times \Gamma_0 \times \Gamma_0 \longrightarrow [0,\infty]$ is jointly Borel-measurable and satisfies
 \begin{align}\label{KE}
  \int \limits_{\Gamma}\int \limits_{\Gamma_0}K(\gamma,\xi,\eta,\zeta) d\lambda(\zeta)d\mu(\gamma) < \infty, \ \ \forall \xi, \eta \in \Gamma_0.
 \end{align}
\end{enumerate}
By \eqref{KE} one immediately shows that \eqref{EQ:02} is well-defined for any $\eta \in \Gamma_0$ and $\mu$-a.a. $\gamma \in \Gamma$.
However, since we have not assumed any growth condition on the integral in \eqref{KE}, the resulting function $L^SF$ does not need to be bounded.
In particular, the corresponding dynamics may be not conservative, see Section 2 for additional comments.
Particular examples are discussed in the last section of this work, see also \cite{FRIESEN15}.

\subsection{The main result}
As it is already explained in the introduction, we are interested in the asymptotic regime $\e \to 0$ for the densities $\rho_t^{\e}$ obtained from \eqref{EQ:01}.
However, in this generality it seems hopeless to study the Fokker-Planck equation \eqref{EQ:01} directly.
For this purpose we introduce a certain approximation $L^S_{\delta}$ and study first the corresponding limit $\e \to 0$ when $\delta > 0$ is fixed.
Afterwards we take the limit $\delta \to 0$ to deduce the desired result.
Below we briefly introduce the main objects of this work. Their properties are studied afterwards.
\begin{enumerate}
 \item[(i)] For given $\delta > 0$ we define
\[
 K_{\delta}(\gamma,\xi,\eta, \zeta) := e^{-\delta q(\gamma,\eta)}K(\gamma,\xi,\eta, \zeta), \qquad
 q(\gamma,\eta) = \sum \limits_{\xi \subset \eta}\int \limits_{\Gamma_0}K(\gamma, \xi,\eta,\zeta)d\lambda(\zeta).
\]
Then $K_{\delta}$ and $q$ are measurable, non-negative and, by \eqref{KE}, also finite for $\mu$-a.a. $\gamma$.
Denote by $L_{\delta}^S$ the operator given by \eqref{EQ:02} with $K$ replaced by $K_{\delta}$
 and define another operator on $\mathcal{L}_{\mu}$ by
\begin{align*}
 (L_{\delta}^S)^*\rho(\eta,\gamma) &= - q(\gamma,\eta)e^{-\delta q(\gamma,\eta)}\rho(\eta,\gamma) 
 \\ &\ \ \ + \sum \limits_{\xi \subset \eta}\int \limits_{\Gamma_0}\rho(\eta \backslash \xi \cup \zeta,\gamma)e^{-\delta q(\gamma, \eta \backslash \xi \cup \zeta)}K(\gamma, \zeta, \eta \backslash \xi \cup \zeta, \xi)d\lambda(\zeta).
\end{align*}
\item[(ii)] Let us fix the notation for the limiting objects when $\e \to 0$, i.e. define
\begin{align}\label{AVERAGEDKERNEL2}
 \overline{K}_{\delta}(\xi,\eta,\zeta) &:= \int \limits_{\Gamma}K_{\delta}(\gamma, \xi, \eta, \zeta)d\mu(\gamma),
  \ \ \overline{q}_{\delta}(\eta) = \sum \limits_{\xi \subset \eta}\int \limits_{\Gamma_0}\overline{K}_{\delta}(\xi, \eta,\zeta)d\lambda(\zeta)
\end{align}
and associated to $\overline{K}_{\delta}$ consider the Markov (pre-)generator
\begin{align*}
 \overline{L}_{\delta}F(\eta) &= \sum \limits_{\xi \subset \eta} \int \limits_{\Gamma_0}(F(\eta \backslash \xi \cup \zeta) - F(\eta))\overline{K}_{\delta}(\xi, \eta, \zeta)d\lambda(\zeta).
\end{align*}
Finally define another operator on $L^1(\Gamma_0,\lambda)$ by
\begin{align*}
 \overline{L}_{\delta}^*\rho(\eta) &= - \overline{q}_{\delta}(\eta)\rho(\eta)
  + \sum \limits_{\xi \subset \eta}\int \limits_{\Gamma_0}\rho(\eta \backslash \xi \cup \zeta)\overline{K}_{\delta}(\zeta, \eta \backslash \xi \cup \zeta,\xi)d\lambda(\zeta).
\end{align*}
 \item[(iii)] Finally let us describe the limiting objects when $\delta \to 0$. Define 
\begin{align}\label{AVERAGEDKERNEL1}
 \overline{K}(\xi,\eta,\zeta) &:= \int \limits_{\Gamma}K(\gamma, \xi, \eta, \zeta)d\mu(\gamma), \ \ \overline{q}(\eta) = \sum \limits_{\xi \subset \eta}\int \limits_{\Gamma_0}\overline{K}(\xi, \eta,\zeta)d\lambda(\zeta)
\end{align}
and the associated Markov (pre-)generator
\begin{align*}
 \overline{L}F(\eta) &= \sum \limits_{\xi \subset \eta}\int \limits_{\Gamma_0}(F(\eta \backslash \xi \cup \zeta) - F(\eta))\overline{K}(\xi, \eta, \zeta)d\lambda(\zeta).
\end{align*}
 Finally define another operator on $L^1(\Gamma_0,\lambda)$ by 
\begin{align*}
 \overline{L}^*\rho(\eta) &= - \overline{q}(\eta)\rho(\eta)
  + \sum \limits_{\xi \subset \eta}\int \limits_{\Gamma_0}\rho(\eta \backslash \xi \cup \zeta)\overline{K}(\zeta, \eta \backslash \xi \cup \zeta,\xi)d\lambda(\zeta).
\end{align*}
\end{enumerate}
Below we summarize the main properties of these operators.
\begin{Lemma}\label{LEMMA:01}
 Suppose that (S) and (E) are satisfied. Then
 \begin{enumerate}
  \item[(a)] For each $\delta > 0$, $L_{\delta}^S$ is bounded on $L^{\infty}(\Gamma_0 \times \Gamma, \lambda \otimes \mu)$
 and $(L_{\delta}^S)^*$ is bounded on $\mathcal{L}_{\mu}$. Moreover, for each $F \in L^{\infty}(\Gamma_0\times \Gamma,\lambda \otimes \mu)$
 and $\rho \in \mathcal{L}_{\mu}$, one has
\begin{align}\label{EQ:08}
 \int \limits_{\Gamma_0 \times \Gamma} (L_{\delta}^SF)(\eta,\gamma)\rho(\eta,\gamma) d\lambda(\eta)d\mu(\gamma)
 = \int \limits_{\Gamma_0 \times \Gamma}F(\eta,\gamma)(L_{\delta}^S)^*\rho(\eta,\gamma)d\lambda(\eta)d\mu(\gamma).
\end{align}
 \item[(b)] For each $\delta > 0$, $\overline{L}_{\delta}$ is a bounded on $L^{\infty}(\Gamma_0,\lambda)$
 and $\overline{L}_{\delta}^*$ is bounded on $L^1(\Gamma_0,\lambda)$. Moreover, for each $F \in L^{\infty}(\Gamma_0,\lambda)$
 and $\rho \in L^1(\Gamma_0,\mu)$, one has
 \[
  \int \limits_{\Gamma_0}(\overline{L}_{\delta}F)(\eta) \rho(\eta)d\lambda(\eta) 
 = \int \limits_{\Gamma_0}F(\eta) (\overline{L}_{\delta}^*\rho)(\eta)d\lambda(\eta).
 \]
 \item[(c)] The operator $\overline{L}^*$ is well-defined on the domain
 \[
  D(\overline{L}^*) = \{ \rho \in L^1(\Gamma_0,\lambda) \ | \ \overline{q}\rho \in L^1(\Gamma_0,\lambda) \},
 \]
 and for all $F \in L^{\infty}(\Gamma_0, \lambda)$ and $\rho \in D(\overline{L}^*)$ it holds that
 \[
  \int \limits_{\Gamma_0}(\overline{L}F)(\eta)\rho(\eta)d\lambda(\eta)
 = \int \limits_{\Gamma_0}F(\eta)(\overline{L}^*\rho)(\eta)d\lambda(\eta).
 \]
 Moreover, the operator $(\overline{L}^*, D(\overline{L}^*))$ has an extension $(\mathcal{G},D(\mathcal{G}))$ on $L^1(\Gamma_0,\lambda)$ which is the generator of a sub-stochastic semigroup on $L^1(\Gamma_0,\lambda)$.
 \end{enumerate}
\end{Lemma}
\begin{proof}
 Let us first prove assertion \textit{(a)}. Take $F \in L^{\infty}(\Gamma_0 \times \Gamma, \lambda \otimes \mu)$, then
\begin{align*}
 |L_{\delta}^SF(\eta,\gamma)|
 &\leq 2 \| F \|_{ L^{\infty}(\Gamma_0 \times \Gamma, \lambda \otimes \mu)} \sum \limits_{\xi \subset \eta} \int \limits_{\Gamma_0}K_{\delta}(\gamma,\xi,\eta,\zeta)d\lambda(\zeta)
 \\ &=  2 \| F \|_{ L^{\infty}(\Gamma_0 \times \Gamma, \lambda \otimes \mu)} q(\gamma,\eta)e^{- \delta q(\gamma,\eta)}
 \\ &\leq  \| F \|_{ L^{\infty}(\Gamma_0 \times \Gamma, \lambda \otimes \mu)} \frac{2}{e \delta},
\end{align*}
 i.e. $L_{\delta}^S$ is bounded. For $(L_{\delta}^S)^*$ we apply twice \eqref{IBP} to deduce
\begin{align*}
 \| (L_{\delta}^S)\rho \|_{\mathcal{L}_{\mu}} &\leq \frac{1}{e\delta} \| \rho \|_{\mathcal{L}_{\mu}} + \int \limits_{\Gamma}  \int \limits_{\Gamma_0} \int \limits_{\Gamma_0}\int \limits_{\Gamma_0} |\rho(\eta \cup \zeta, \gamma)| e^{- \delta q(\gamma, \eta \cup \zeta)} K(\gamma, \zeta, \eta \cup \zeta, \xi)d\lambda(\zeta)d\lambda(\xi)d\lambda(\eta)d\mu(\gamma)
\\ &=  \frac{1}{e\delta} \| \rho \|_{\mathcal{L}_{\mu}} + \int \limits_{\Gamma} \int \limits_{\Gamma_0} \int \limits_{\Gamma_0}\sum \limits_{\zeta \subset \eta} |\rho(\eta, \gamma)| e^{- \delta q(\gamma, \eta)} K(\gamma, \zeta, \eta, \xi)d\lambda(\xi)d\lambda(\eta)d\mu(\gamma)
 \\ &=  \frac{1}{e\delta} \| \rho \|_{\mathcal{L}_{\mu}} + \int \limits_{\Gamma} \int \limits_{\Gamma_0} |\rho(\eta, \gamma)|e^{- \delta q(\gamma,\eta)} q(\gamma,\eta) d\lambda(\eta) d\mu(\gamma)
 \\ &\leq \| \rho \|_{\mathcal{L}_{\mu}} \frac{2}{e\delta}.
 \end{align*}
 Identity \eqref{EQ:08} follows by a very similar computation using \eqref{IBP}, details are left for the reader.
 Assertion \textit{(b)} can be shown in exactly the same way, while assertion \textit{(c)} is a consequence of Section 2.
\end{proof}
It is worthwhile to mention that $(\mathcal{G},D(\mathcal{G}))$ does not need to be the closure of $(\overline{L}^*,D(\overline{L}^*))$. 
A characterization and sufficient conditions for this property are given in Section 2.

In order to study the Fokker-Planck equation for the joint evolution of scaled densities, 
we have first to extend the semigroup $T^E(t)$ onto $\mathcal{L}_{\mu}$. Define 
\[
 D = \left\{ f = \sum \limits_{k=1}^{n}R_k \rho_k \ \bigg | \ n \in \mathbb{N}, R_k \in L^1(\Gamma,\mu), \rho_k \in L^1(\Gamma_0, \lambda)\right\} \subset \mathcal{L}_{\mu}.
\]
Note that $D \subset \mathcal{L}_{\mu}$ is dense, see \cite[Proposition 5.5.6]{GRAFAKOS04}.
The next lemma shows that $T^E(t)$ given by assumption (E) can be uniquely extended to an ergodic semigroup on $\mathcal{L}_{\mu}$.
\begin{Lemma}\label{LEMMA:00}
 Suppose that condition (E) is satisfied. Then there exists a unique positive contraction semigroup $\widetilde{T}^E(t)$ on $\mathcal{L}_{\mu}$ such that
 \begin{align}\label{EQ:07}
  \widetilde{T}^E(t)f = \sum \limits_{k=1}^{n}(T^E(t)R_k)\rho_k, \ \ f \in D.
 \end{align}
 Moreover, it holds that
 \[
  \Vert \widetilde{T}^E(t)f - \widetilde{P}_{\mu}f \Vert_{\mathcal{L}_{\mu}} \longrightarrow 0, \ \ t \to \infty,
 \]
 where $\widetilde{P}_{\mu}f := \int_{\Gamma}f(\cdot, \gamma)d\mu(\gamma) \in L^1(\Gamma_0,\lambda)$ denotes the averaging with respect to $\mu$.
 Let $(L^E,D(L^E))$ be the generator of $T^E(t)$, $(\widetilde{L}^E, D( \widetilde{L}^E))$ be the generator of $\widetilde{T}^E(t)$, and define
 \[
  \mathcal{D} = \left\{ f = \sum \limits_{k=1}^{n}R_k \rho_k \ \bigg | \ n \in \mathbb{N}, R_k \in D(L^E), \rho_k \in L^1(\Gamma_0, \lambda)\right\}.
 \]
 Then $\mathcal{D}$ is a core for the generator $(\widetilde{L}^E, D( \widetilde{L}^E))$ such that
 \[
  \widetilde{L}^Ef = \sum \limits_{k=1}^{n}\rho_k L^ER_k, \ \ f \in \mathcal{D}.
 \]
\end{Lemma}
\begin{proof}
 First observe that $\mathcal{L}_{\mu} := L^1(\Gamma_0 \times \Gamma, \mu \otimes \lambda) \cong L^1(\Gamma \to L^1(\Gamma_0,\lambda),\mu)$.
 Define $\widetilde{T}^E(t)$ by \eqref{EQ:07}. Since $T^E(t)$ is positive, $\widetilde{T}^E(t)$ has a bounded extension to $\mathcal{L}_{\mu}$
 with the same norm, see \cite[Proposition 5.5.10]{GRAFAKOS04}. In particular $\widetilde{T}^E(t)$ is a contraction operator.
 Since $\widetilde{T}^E(t)$ is strongly continuous on $D$, it follows that it is also strongly continuous on all $\mathcal{L}_{\mu}$.
 Let us prove the ergodicity. Take $f \in D$, then
\[
 \Vert \widetilde{T}^E(t)f - \widetilde{P}_{\mu}f \Vert_{\mathcal{L}_{\mu}} \leq \sum \limits_{k=1}^{n}\Vert T^E(t)R_k - P_{\mu}R_k \Vert_{L^1(\Gamma,\mu)}\Vert \rho_k\Vert_{L^1(\Gamma_0,\lambda)} \longrightarrow 0, \ \ t \to \infty.
\]
 Since $\widetilde{T}^E(t)$ is a semigroup of contractions and $D$ dense in $\mathcal{L}_{\mu}$, the assertion is proved.
 For the last assertion observe that $\mathcal{D}$ is dense in $\mathcal{L}_{\mu}$.
 Moreover, by \eqref{EQ:07} it follows that $\mathcal{D}$ is also invariant for $\widetilde{T}^E(t)$ and hence it is a core.
\end{proof}
The following is our main result of this work.
\begin{Theorem}\label{AVERAGINGTHEOREM}
 Assume that conditions (E) and (S) are satisfied. Then
 \begin{enumerate}
  \item[(a)] For any $\e > 0$ and $\delta > 0$, the operator $(L_{\delta}^S)^* + \frac{1}{\e}\widetilde{L}^E$ equipped with the domain $\mathcal{D}$
 is closable and its closure is the generator of a stochastic semigroup $T^{\e,\delta}(t)$ on $\mathcal{L}_{\mu}$.
 \item[(b)] For any $\delta > 0$ and any $\rho \in L^1(\Gamma_0,\lambda) \subset \mathcal{L}_{\mu}$ one has
 \begin{align}\label{AVERAGING}
   \lim \limits_{\e \to 0} \sup\limits_{t \in [0,T]} \| T^{\e,\delta}(t)\rho - e^{t \overline{L}_{\delta}^*}\rho\|_{\mathcal{L}_{\mu}} = 0, \ \ \forall T > 0.
 \end{align}
 \item[(c)] Suppose that $(\overline{L}^*, D(\overline{L}^*))$ is closable and its closure generates a stochastic semigroup $\overline{T}(t)$ on $L^1(\Gamma_0,\lambda)$. Then
 \begin{align}\label{DELTA}
  \lim \limits_{\delta \to 0}\sup \limits_{t \in [0,T]}\| e^{t \overline{L}_{\delta}^*}\rho - \overline{T}(t)\rho \|_{L^1(\Gamma_0,\lambda)} = 0, \ \ \forall T > 0, 
 \ \ \rho \in L^1(\Gamma_0,\lambda).
 \end{align}
 \end{enumerate}
\end{Theorem}
\begin{proof}
 \textit{(a)} First observe that $\frac{1}{\e}\widetilde{L}^E$ is, for any $\e > 0$, the generator of the semigroup $\widetilde{T}^E(\frac{t}{\e})$ on $\mathcal{L}_{\mu}$. 
 Moreover, $\mathcal{D}$ is a core for this generator.
 Since $(L_{\delta}^S)^*$ is bounded on $\mathcal{L}_{\mu}$, the sum $(L_{\delta}^S)^* + \frac{1}{\e}\widetilde{L}^E$ is defined on $\mathcal{D}$,
 it is closable and the closure generates a semigroup $T_{\e,\delta}(t)$ on $\mathcal{L}_{\mu}$.
 Due to the Trotter product formula this semigroup is sub-stochastic.

 \textit{(b)} The assertion is proved if we can show that \cite[Theorem 2.1]{KURTZ73} is applicable. 
 Therefore observe that for $\rho \in \mathcal{L}_{\mu}$ and $\lambda > 0$
 \[
  \left \Vert \lambda \int \limits_{0}^{\infty}e^{-\lambda t}\widetilde{T}^E(t)\rho dt - \widetilde{P}_{\mu}\rho \right\Vert_{\mathcal{L}_{\mu}} \leq \int \limits_{0}^{\infty}e^{-s}\left \Vert \widetilde{T}^E\left(\frac{s}{\lambda}\right)\rho - \widetilde{P}_{\mu}\rho \right \Vert_{\mathcal{L}_{\mu}}ds.
 \]
 Since $\widetilde{T}^E(t)$ is ergodic on $\mathcal{L}_{\mu}$ it follows that, for fixed $s \geq 0$, the integrand tends to zero as $\lambda \to 0$. 
 Due to $\| \widetilde{P}_{\mu}\rho \|_{\mathcal{L}_{\mu}} \leq \Vert \rho \Vert_{\mathcal{L}_{\mu}}$
 and $\| \widetilde{T}^E(t) \rho \|_{\mathcal{L}_{\mu}} \leq \| \rho \|_{\mathcal{L}_{\mu}}$ the integrand is bounded by $2\Vert \rho \Vert_{\mathcal{L}_{\mu}}e^{-s}$ and hence dominated convergence implies 
 \[
  \widetilde{P}_{\mu}\rho = \underset{\lambda \to 0}{\lim}\ \lambda \int \limits_{0}^{\infty}e^{-\lambda t}\widetilde{T}^E(t)\rho dt, \ \ \rho \in \mathcal{L}_{\mu}.
 \]
 The operator $\widetilde{P}_{\mu}$ is a projection on $\mathcal{L}_{\mu}$ with range $\mathrm{Ran}(\widetilde{P}_{\mu}) \cong L^1(\Gamma_0, \lambda)$.
 Following the notion of \cite{KURTZ73}, observe that $C\rho := \widetilde{P}_{\mu}(L_{\delta}^S)^*\rho = \overline{L}_{\delta}^*\rho$ is defined on $L^1(\Gamma_0,\lambda)$ and is additionally bounded. Hence \cite[Theorem 2.1]{KURTZ73} is applicable, i.e. \eqref{AVERAGING} is proved.

 \textit{(c)} For the last assertion observe that $D(\widetilde{L}^*)$ is a core for $\overline{T}(t)$. 
 By Trotter-Kato approximation it suffies to show that
 \[
  \| \overline{L}_{\delta}^* \rho - \overline{L}^* \rho \|_{L^1(\Gamma_0,\lambda)} \longrightarrow 0, \ \ \delta \to 0, \ \ \rho \in D(\overline{L}^*).
 \]
 Indeed, for any $\rho \in D(\overline{L}^*)$, we obtain
 \begin{align*}
  & \Vert \overline{L}_{\delta}^*\rho - \overline{L}^*\rho\Vert 
  \\ &\leq \int \limits_{\Gamma_0}|\rho(\eta)| |\overline{q}_{\delta}(\eta) - \overline{q}(\eta)|d\lambda(\eta)
  \\ & \ \ + \int \limits_{\Gamma_0}\sum \limits_{\xi \subset \eta}\int \limits_{\Gamma_0}|\rho(\eta \backslash \xi \cup \zeta)| |\overline{K_{\delta}}(\zeta, \eta \backslash \xi \cup \zeta, \zeta) - \overline{K}(\zeta, \eta \backslash \xi \cup \zeta, \xi)|d\lambda(\zeta)d\lambda(\eta).
 \end{align*}
 For the first term observe that by \eqref{AVERAGEDKERNEL1} and \eqref{AVERAGEDKERNEL2} we obtain
 \begin{align*}
  |\overline{q}_{\delta}(\eta) - \overline{q}(\eta)| &\leq \sum \limits_{\xi \subset \eta}\int \limits_{\Gamma_0}|\overline{K}_{\delta}(\xi, \eta, \zeta) - \overline{K}(\xi, \eta, \zeta)|d \lambda(\zeta)
 \\ &\leq \sum \limits_{\xi \subset \eta} \int \limits_{\Gamma_0}\int \limits_{\Gamma}\left| 1 - e^{-\delta q(\gamma,\eta)}\right| K(\gamma, \xi,\eta,\zeta) d\mu(\gamma) d\lambda(\zeta).
 \end{align*}
 Since the integrand tends pointwise to zero and is bounded by $2 K(\gamma, \xi,\eta,\zeta)$,
 we deduce from dominated convergence
 \[
  \int \limits_{\Gamma_0}|\rho(\eta)| |\overline{q}_{\delta}(\eta) - \overline{q}(\eta)|d\lambda(\eta) \longrightarrow 0, \ \ \delta \to 0.
 \]
 The convergence 
 \[
 \int \limits_{\Gamma_0}\sum \limits_{\xi \subset \eta}\int \limits_{\Gamma_0}|\rho(\eta \backslash \xi \cup \zeta)| |\overline{K_{\delta}}(\zeta, \eta \backslash \xi  \cup \zeta, \zeta) - \overline{K}(\zeta, \eta \backslash \xi \cup \zeta, \xi)|d\lambda(\zeta)d\lambda(\eta) \longrightarrow 0, \ \ \delta \to 0
 \]
 can be shown in the same way.
\end{proof}

\section{Examples}
Consider equilibrium diffusions or Glauber birth-and-death Markov dynamics on $\Gamma$ for a given invariant (Gibbs) measure $\mu$,
more generally suppose that condition (E) is satisfied. Let us consider the spatial logistic model with heuristic Markov generator
\begin{align*}
 (L^SF)(\eta,\gamma) &= \sum \limits_{x \in \eta}\left( m(x,\gamma) + \sum \limits_{y \in \eta \backslash x}a^-(x-y)\right)(F(\eta \backslash x,\gamma) - F(\eta,\gamma))
 \\ & \ \ \ + \sum \limits_{x \in \eta} \lambda(x,\gamma)\int \limits_{\mathbb{R}^d}a^+(x-y)(F(\eta \cup y,\gamma) - F(\eta,\gamma))dy.
\end{align*}
The statistical dynamics for such model (without the presence of an environment) has been analyzed, e.g.,  in \cite{FKK08, FKK13, FKKK15, KK18}.
Here $m \geq 0$ is the intensity of the death of particles and $\lambda \geq 0$ describes fecundity effects caused by the environment in the state $\gamma$.
Finally $a^- \geq 0$ is assumed to be symmetric. It describes the competition of particles from the configuration $\eta \in \Gamma_0$.
The distribution of new particles is described by a symmetric probability density $a^+$ on $\mathbb{R}^d$.
After scaling the averaged dynamics will be given by the heuristic Markov operator
\begin{align*}
 (\overline{L}F)(\eta) &= \sum \limits_{x \in \eta}\left( \overline{m}(x) + \sum \limits_{y \in \eta \backslash x}a^-(x-y)\right)(F(\eta \backslash x) - F(\eta))
 \\ & \ \ \ + \sum \limits_{x \in \eta} \overline{\lambda}(x)\int \limits_{\mathbb{R}^d}a^+(x-y)(F(\eta \cup y) - F(\eta))dy,
\end{align*}
where the averaged intensities are given by
\begin{align*}
 \overline{m}(x) = \int \limits_{\Gamma}m(x,\gamma)d\mu(\gamma), \qquad
 \overline{\lambda}(x) = \int \limits_{\Gamma}\lambda(x,\gamma)d\mu(\gamma).
\end{align*}
Proceeding as in Section 3, denote by $T_{\e,\delta}(t)$ the scaled semigroup on densities $\mathcal{L}_{\mu}$ and by
$\overline{T}(t)$ and $e^{t \overline{L}_{\delta}^*}$ the semigroups on $L^1(\Gamma_0,\lambda)$.
The next result states conditions for which these semigroups exist and \eqref{AVERAGING} holds.
\begin{Theorem}
 Assume that all intensities $a^{\pm}, m, \lambda$ are non-negative, measurable, that $a^+$ is a probability density and that
 $m(x,\cdot), \lambda(x,\cdot)$ are integrable with respect to $\mu$ for any $x \in \mathbb{R}^d$. 
 Then the semigroups $T_{\e,\delta}(t)$, $e^{t \overline{L}_{\delta}^*}$ and $\overline{T}(t)$ exist and \eqref{AVERAGING} holds.
\end{Theorem}
\begin{proof}
 First observe that $\eta \in \Gamma_0$ and fixed $\xi \subset \eta$
 \begin{align*}
  q(\gamma,\eta) &= \sum \limits_{x \in \eta}m(x,\gamma) + \sum \limits_{x \in \eta}\sum \limits_{y \in \eta \backslash x} a^-(x-y) + \sum \limits_{x \in \eta}\lambda(x,\gamma)
  \\ &= \sum \limits_{\xi \subset \eta}\int \limits_{\Gamma_0}K(\gamma,\xi,\eta,\zeta)d\lambda(\zeta)
  \geq \int \limits_{\Gamma_0}K(\gamma,\xi,\eta,\zeta) d\lambda(\zeta).
 \end{align*}
 This implies that 
 \[
  \int \limits_{\Gamma}\int \limits_{\Gamma_0}K(\gamma,\xi,\eta, \zeta)d\lambda(\zeta)d\mu(\gamma)
  \leq \int \limits_{\Gamma}q(\gamma,\eta)d\mu(\gamma) < \infty,
 \]
 i.e. condition (S) are satisfied.
 The assertion is a direct consequence of Theorem \ref{AVERAGINGTHEOREM}.(b).
\end{proof}
The reader may wonder why such weak assumptions are sufficient for existence and convergence of the semigroups. 
The crucial point here is that we consider an approximation
by bounded linear operators and hence, for each $\delta > 0$, no additional conditions are needed. 
In order to pass to the limit $\delta \to 0$ additional assumptions are necessary, which are given below.
\begin{Theorem}
 Assume that the conditions of previous theorem are fulfilled. Moreover suppose that
 \begin{enumerate}
  \item[(i)] either $\overline{m}, \overline{\lambda},a^-$ are bounded
  \item[(ii)] or $\overline{m}, \overline{\lambda}, a^-$ are locally bounded and 
  there exists a continuous function $\varphi: \mathbb{R}^d \longrightarrow [1,\infty)$ with $\varphi(x) \longrightarrow \infty$ when $|x| \to \infty$
  and a constant $c > 0$ such that
 \begin{align}\label{LYAPUNOV}
  \overline{\lambda}(x)(a^+ \ast \varphi)(x) \leq  c \varphi(x) + \varphi(x)\overline{m}(x), \ \ x \in \mathbb{R}^d.
 \end{align}
 \end{enumerate}
 Then $\overline{T}(t)$ is stochastic and \eqref{DELTA} holds.
\end{Theorem}
\begin{proof}
 In the first case set $E_n = \{ \eta \in \Gamma_0 \ | \ |\eta| \leq n \}$, then $E_n \subset E_{n+1}$, $\bigcup_{n \geq 1}E_n = \Gamma_0$ and
 \[
  \overline{q}(\eta) = \sum \limits_{x \in \eta} \overline{m}(x) + \sum \limits_{x \in \eta}\overline{\lambda}(x) + \sum \limits_{x \in \eta} \sum \limits_{y \in \eta\backslash x}a^-(x-y)
 \]
 is bounded on any $E_n$. Moreover, for $V(\eta) = |\eta|$ we obtain $\inf_{\eta \not \in E_n}V(\eta) \geq n+1 \to \infty$, $n \to \infty$ and hence
 the assertion follows from Remark \ref{REMARK:00}. 

 For the second case take $E_n = \{ \eta \in \Gamma_0 \ | \ |\eta| \leq n, \ \eta \subset B_n\}$,
 where $B_n \subset \mathbbm{R}^d$ is a ball centered at zero of radius $n$. 
 Hence due to \eqref{LYAPUNOV} we see that the Lyapunov function $V(\eta) = \sum_{x \in \eta}\varphi(x)$ satisfies
 \[
  (\overline{L}V)(\eta) \leq c V(\eta), \ \ \eta \in \Gamma_0.
 \]
 The assertion follows again by Remark \ref{REMARK:00}.
\end{proof}
As a concrete case we can take $\mu = \pi_{z}$, that is the Poisson measure with intensity $z > 0$. Let us take for the interactions
\[
 m(x,\gamma) = m_0 + \sum \limits_{y \in \gamma}\kappa(x-y)
\]
and
\[
 \lambda(x,\gamma) = \lambda_0 + \sum \limits_{y \in \gamma}\psi(x-y)
\]
with $\lambda_0 > m_0$, $0 \leq \kappa,\psi \in L^1(\mathbb{R}^d)$ and $\langle \psi\rangle < \langle \kappa \rangle$. Then 
\[
 \overline{m} = m_0 + z \int \limits_{\mathbb{R}^d}\kappa(y)dy = m_0 + z \langle \kappa \rangle
\]
and 
\[
 \overline{\lambda} = \lambda_0 + z \int \limits_{\mathbb{R}^d}\psi(y)dy = \lambda_0 + \langle \psi \rangle.
\]
Define $\beta(z) = (\lambda_0 + z \langle \psi\rangle - m_0 - z \langle \kappa \rangle)$ and observe that $V(\eta) = 1 + |\eta|$ satisfies
\[
 (\overline{L}V)(\eta) \leq \beta(z)|\eta|.
\]
If $a^-$ is, in addition, bounded, then for each $0 \leq \rho \in L^1(\Gamma_0, \lambda)$ satisfying
\[
 \int \limits_{\Gamma_0}(1+|\eta|)\rho(\eta)d\lambda(\eta) < \infty, \qquad 
 \int \limits_{\Gamma_0}\rho(\eta)d\lambda(\eta) = 1
\]
we see that $\rho_t = \overline{T}(t)\rho$ satisfies
\[
  \int \limits_{\Gamma_0}|\eta|\rho_t(\eta) d\lambda(\eta) \leq e^{\beta(z)t}\int \limits_{\Gamma_0}|\eta|\rho(\eta)d\lambda(\eta), \ \ t \geq 0.
\]
Without the presence of an environment, i.e. $z = 0$, the number of particles within the system will grow exponentially in time. 
But due to the influence of the environment, such growth may be prevented
or even exponential decay may be observed. 

\subsection*{Acknowledgments}
The authors gratefully acknowledge the useful remarks of the anonymous referee
which lead to an improvement of this work.

\begin{footnotesize}

\bibliographystyle{alpha}
\bibliography{Bibliography}

\end{footnotesize}

\end{document}